\documentclass[a4paper,10pt,onecolumn]{IEEEtran}
\usepackage{mathtools}
\usepackage{array}
\usepackage{amssymb}
\usepackage{amsmath}
\usepackage{amsthm}
\usepackage{color}
\usepackage{enumerate}
\usepackage{graphicx}
\usepackage{bbm}

\newlength{\eqboxstorage}

\newcommand\cA{\mathcal{A}}

\newcommand{\Ber}{\mathrm{Ber}}
\newcommand{\Berhalf}{\Ber(1/2)}

\newcommand{\BN}{\mathsf{B}_N}
\newcommand{\GN}{\mathsf{G}_N}
\newcommand{\BEight}{\mathsf{B}_8}
\newcommand{\GEight}{\mathsf{G}_8}

\newcommand{\HU}{H(\tilde{U}_i | \underline{Q_i})}
\newcommand{\HV}{H(\tilde{V}_i | \underline{R_i})}

\newcommand{\Arikan}{Ar\i kan}

\newtheorem{theo}{Theorem}
\newtheorem{conj}[theo]{Conjecture}
\newtheorem{lemm}[theo]{Lemma}
\newtheorem{coro}[theo]{Corollary}

\theoremstyle{definition}
\newtheorem{defi}{Definition}

\newcommand{\twobibs}[2]{#2} 

\title{Polar Coding for Processes with Memory}

\IEEEoverridecommandlockouts

\author{%
\IEEEauthorblockN{\bf Eren \c Sa\c so\u glu}
\IEEEauthorblockA{
Apple Inc.\thanks{This work was done when Eren \c Sa\c so\u glu was at the
Technion in June--July 2015. It was presented in part at ISIT 2016.}\\ 
Cupertino, CA, USA\\
{\tt eren.sasoglu@gmail.com}\\}
\and
\IEEEauthorblockN{\bf Ido Tal}
\IEEEauthorblockA{
Department of Electrical Engineering\\
Techion, Haifa 32000, Israel\\
{\tt idotal@ee.technion.ac.il}\\}
}

\begin{document}
\maketitle
\begin{abstract}
    We study polar coding for stochastic processes with memory. For example, a process may be defined by the joint distribution of the input and output of a channel. The memory may be present in the channel, the input, or both. We show that $\psi$-mixing processes polarize under the standard \Arikan\ transform, under a mild condition. We further show that the rate of polarization of the \emph{low-entropy} synthetic channels is roughly $O(2^{-\sqrt{N}})$,  where $N$ is the blocklength. That is, essentially the same rate as in the memoryless case. 
\end{abstract}

\section{Introduction}

Polar codes were invented by \Arikan~\cite{Arikan2009} as a
low-complexity method to achieve the capacity of symmetric
binary-input memoryless channels.  The technique that underlies these
codes, called \emph{polarization}, is quite versatile, and has since
been applied to numerous classical memoryless problems in information
theory.  

Many practical sources and channels are not well-described by
memoryless models.  In wireless communication, for example, memory in
the form of intersymbol interference is quite prominent due to
multipath propagation, as are slow variations in channel conditions
due to mobility.  In practice, this type of memory is commonly
handled by eliminating it, e.g., by augmenting the transmitter/receiver
appropriately to create an overall memoryless channel.  Memoryless
coding techniques are then used for communication.  Channel
equalization, interleaving, and OFDM techniques are perhaps the most notable examples
of this approach.  

In contrast, we are interested here in whether polar coding can be
used \emph{directly} on channels and sources with memory.  In addition
to being of theoretical interest, such results may help simplify the
design of communication or compression systems.  

Little is known about the theory of polarization for settings with
memory.  In particular, it was shown in~\cite{Wang+:15c} that the
successive cancellation decoding complexity of polar codes scales with
the number of states of the underlying process, and thus is practical
if the amount of memory in the system is modest.  It was shown
in~\cite[Chapter 5]{Sasoglu2011} that \Arikan's standard transform
indeed polarizes a class of mixing processes with finite memory.
Whether polarization takes place sufficiently fast to yield a coding
theorem has been left open, however, and that is the problem we address
here.

We first give a proof of polarization that is both simpler than the
one given in~\cite{Sasoglu2011}, and holds for the more general class
of $\psi$-mixing processes with finite $\psi_0$ (both concepts are defined in Section~\ref{sec:setting}). We further show that the asymptotic rate of
polarization of the \emph{low-entropy} synthetic channels is as in the memoryless
case. Conversely, we show a simple counter-example of a process that is not $\psi$-mixing and which does not polarize because it is periodic. We remark that in \cite{ShuvalTal:17.2p}, under additional assumptions, fast polarization is shown for the \emph{high-entropy} synthetic channels. 


\section{Setting}
\label{sec:setting}
Let $(X_i,Y_i)$, $i\in\mathbb{Z}$, be a stationary process, where
the $Y_i$ take values in a finite alphabet $\mathcal{Y}$.  We assume $X_i\in\{0,1\}$ 
to keep the notation simple, but the results here can be generalized
to arbitrary finite alphabets using standard techniques.  See, for
example,~\cite[Chapter 3]{Sasoglu2011}. We think of $X_i$ as a sequence to be estimated, and~$Y_i$ as a
sequence of observations related to~$X_i$.  In particular,~$X_i$ may
be the input sequence to a communication channel, with the
corresponding channel output~$Y_i$.  Alternatively,~$X_i$ may be the
output of a data source to be compressed, and~$Y_i$ may be the side
information available to the decompressor. 

A key property of the processes we consider is $\psi$-mixing. We follow\footnote{To the best of our understanding, the first displayed equation on page 169 of \cite{Shields1996} should be ``$\sum_v \mu(uvw) \leq \cdots$''.}~\cite[Page 169]{Shields1996} and say that a process~$T_i$ is $\psi$-mixing if there exists a nonincreasing sequence $\psi_k \to  1$ as $k\to\infty$ such that
\begin{align}
\label{eqn:psi-mixing}
\Pr(A\cap B)\le \psi_k \Pr(A)\Pr(B)
\end{align}
for all $A\in\sigma(T_{-\infty}^0)$ and $B\in\sigma(T_{k+1}^{\infty})$,
where $\sigma(\cdot)$ denotes the sigma-field generated by its
argument.  Since $\psi_k \to 1$, in a
$\psi$-mixing process, any two events $A\in\sigma(T_{-\infty}^0)$ and $B\in\sigma(T_{k+1}^{\infty})$
that are sufficiently separated in `time' are almost independent. Namely, by \cite[Definition 3.3, page 67]{Bradley:07b}, \cite[Proposition 3.11, part a, page 76]{Bradley:07b}, and \cite[Proposition 5.2, part III.a, page 153]{Bradley:07b} 
\[
    |\Pr(A \cap B) - \Pr(A)\Pr(B)| \leq \frac{\psi_k-1}{2} \; .
\]

In this paper, we require for polarization that a process be $\psi$-mixing with finite $\psi_0$. Since this requirement appears several times, we make the following definition.

\begin{defi}[Promptly $\psi$-mixing]
    Let $(X_i,Y_i)$, $i\in\mathbb{Z}$, be a stationary process, where  $X_i\in\{0,1\}$ and the $Y_i$ take values in a finite alphabet $\mathcal{Y}$. Such a process is called \emph{promptly $\psi$-mixing} if it is $\psi$-mixing and $\psi_0 < \infty$.
\end{defi}


Many source and channel models of practical importance satisfy our requirements of being promptly $\psi$-mixing. Specifically, this holds for a class of models with memory that have an underlying ergodic Markov structure, as shown in \cite[Lemma 5]{ShuvalTal:17.2p}. There, these processes are termed Finite-state, Aperiodic, Irreducible (hidden) Markov processes, or FAIM for short. The parameter~$\psi_0$ plays an important role in this paper,
and can be computed easily if the underlying process is FAIM \cite[Equation 19]{ShuvalTal:17.2p}.

We are interested in the effects of Ar\i kan's standard polar
transform on stationary processes with memory.  For this purpose, we let
$U_1^N=X_1^N \BN \GN$, where the matrix multiplications are over the
binary field, $N=2^n$ for positive integers $n$, $\GN$ is the $n$th
Kronecker power of
$\big(\begin{smallmatrix}1&0\\1&1\end{smallmatrix}\big)$, and $\BN$ is
the $N \times N$ bit-reversal matrix.  The conditional
entropy rate of $X_i$ is defined as 
\[ 
\mathcal{H}_{X|Y}
	=\lim_{N\to\infty} \frac1N H(X_1^N|Y_1^N) 
	=\lim_{N\to\infty}\frac1N H(X_1^N,Y_1^N)
	-\lim_{N\to\infty}\frac1N H(Y_1^N).
\] 
The limits on the right-hand-side exist due to stationarity
\cite[Theorem 4.2.1]{CoverThomas:06b}.  Also useful for the analysis
is the parameter 
$$
Z(A|B)=2\sum_{b\in\mathcal{B}}
	\sqrt{p_{A,B}(0,b)p_{A,B}(1,b)}
$$
for random variables $A\in\{0,1\}$ and $B \in \mathcal{B}$.  Sometimes
called the Bhattacharyya parameter, $Z(A|B)$ upper-bounds the error
probability of optimally guessing $A$ by observing $B$.  See, for
example, \cite[Proposition~2.2]{Sasoglu2011}. 

\section{Main Results}

The following two theorems relate to the polarization of promptly $\psi$-mixing process.
\begin{theo}[Polarization]
\label{thm:weak}
Let $ (X_i,Y_i)$, $i \in \mathbb{Z}$, be a promptly $\psi$-mixing process, then for all $\epsilon>0$
\begin{align*}
\lim_{N\to\infty}
	\frac1N 
	\big|\big\{
	i: H(U_i|U_1^{i-1},Y_1^N)>1-\epsilon
	\big\}\big|
	= \mathcal{H}_{X|Y} \; ,\\
\lim_{N\to\infty}
	\frac1N 
	\big|\big\{
	i: H(U_i|U_1^{i-1},Y_1^N)<\epsilon
	\big\}\big|
	= 1-\mathcal{H}_{X|Y} \; .
\end{align*}
\end{theo}

\begin{theo}[Fast polarization of the low-entropy set]
\label{thm:strong}
Let $ (X_i,Y_i)$, $i \in \mathbb{Z}$, be a promptly $\psi$-mixing process, then for all $\beta<1/2$
\begin{align*}
\lim_{N\to\infty}
	\frac1N 
	\big|\big\{
	i: Z(U_i|U_1^{i-1},Y_1^N)<2^{-N^\beta}
	\big\}\big|
	= 1-\mathcal{H}_{X|Y} \; .
\end{align*}
\end{theo}

We conjecture that an analog of Theorem~\ref{thm:strong} holds for the high-entropy set. 
\begin{conj}[Fast polarization of the high-entropy set]
\label{conj:strong}
Let $ (X_i,Y_i)$, $i \in \mathbb{Z}$, be a promptly $\psi$-mixing process, then for all $\beta<1/2$
\begin{align*}
\lim_{N\to\infty}
	\frac1N 
	\big|\big\{
	i: Z(U_i|U_1^{i-1},Y_1^N)> 1 - 2^{-N^\beta}
	\big\}\big|
	= \mathcal{H}_{X|Y} \; .
\end{align*}
\end{conj}

Resolving the above conjecture would be an important step for polar codes. We refer the reader to \cite[Theorem 13]{ShuvalTal:17.2p}, which shows that the conjecture indeed holds if the process is FAIM. To recap, assuming that the process $(X_i,Y_i)$ is governed by an underlying state sequence having a certain structure allows one to prove Conjecture~\ref{conj:strong}. However, we will \emph{not} assume an underlying state sequence when proving Theorems~\ref{thm:weak} and \ref{thm:strong}.

As a concrete example of the distinction between promptly $\psi$-mixing and FAIM processes, consider the family of processes  given in \cite[Example 3]{Bradley:99p}. Each such process $(X'_i)$, $i \in \mathbb{Z}$, is $\psi$-mixing, with $\psi_0 < \infty$. Also, the support of each $X'_i$ is $[0,1)$. Next, fix such a process, and let $B$ be some Borel set on $[0,1)$. For example, $B = [0,1/2]$. Define the process $(X_i)$, $i \in \mathbb{Z}$, such that $X_i = 1$ if $X'_i \in B$, and $X_i= 0$ otherwise. Since the process $(X_i)$ is a marginalization of $(X'_i)$, we deduce from (\ref{eqn:psi-mixing}) that $(X_i)$ is also $\psi$-mixing, with finite $\psi_0$. That is, we deduce that $(X_i)$ is promptly $\psi$-mixing, and hence Theorems~\ref{thm:weak} and \ref{thm:strong} are applicable. However, since the underlying process $(X'_i)$ is not finite state, it is not FAIM, and thus it is not clear if Conjecture~\ref{conj:strong} holds for $(X_i)$.

The following theorem shows an example of a process that has memory and that \emph{does not} polarize because it is periodic.
\begin{figure}
\begin{center}
\includegraphics[scale=0.7]{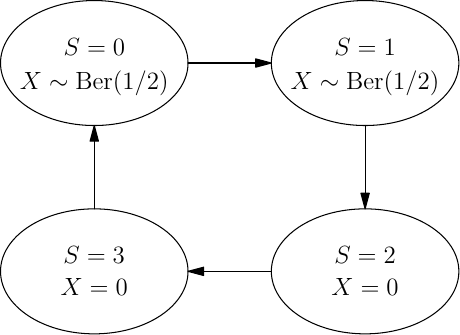}
\end{center}
\caption{A periodic data source that does not polarize.  The source
output is Bernoulli $1/2$ for two consecutive states and zero for next
two consecutive states.  There is no side information, i.e., $Y_i$ is
constant.}
\label{fig:period4}
\end{figure}
\begin{theo}[Periodic processes may not polarize]
\label{thm:periodic}
The stationary periodic Markov process described in
Figure~\ref{fig:period4} does not polarize. Indeed, for all $\frac{5N}{8} < i \leq \frac{6N}{8}$,
\begin{equation}
\left|H(U_i|U_1^{i-1}) - \frac{1}{2}\right| \leq \epsilon_N \; , \quad \lim_{N \to \infty} \epsilon_N = 0 \; .
\end{equation}
\end{theo}

\section{Notation}
We will prove the above theorems in the following sections.  Throughout, we will 
use the shorthand
\begin{align*}
H^\mathbf{b}&=H(U_i|U_1^{i-1},Y_1^N) \; ,\\
Z^\mathbf{b}&=Z(U_i|U_1^{i-1},Y_1^N) \; ,
\end{align*}
where $\mathbf{b}\in\{0,1\}^n$ is the $n$-bit binary expansion of
$i-1\in\{0,\dotsc,N-1\}$.  We will omit the ranges of indices when
they are clear from context.  The following are immediate from 
the definition of $\BN\GN$:
\begin{align*}
H^{\mathbf{b}0}&=H(U_{2i-1}|U_1^{2i-2},Y_1^{2N})\\
H^{\mathbf{b}1}&=H(U_{2i}|U_1^{2i-1},Y_1^{2N})
\end{align*}
for all $\mathbf{b}\in\{0,1\}^n$.  These identities also hold when the
$H$'s are replaced by $Z$'s. Further, if we let $B_1,B_2,\dotsc$ be a sequence of i.i.d.\
$\Berhalf$ random variables, then it is easy to see that the random variables $H_n=H^{B_1\dotsc B_n}$
and $Z_n=Z^{B_1\dotsc B_n}$ are uniformly distributed over the sets of
$H^\mathbf{b}$'s and $Z^\mathbf{b}$'s, respectively.
Theorems
\ref{thm:weak}~and~\ref{thm:strong} are then equivalent to
\begin{theo}
    \label{thm:weakAlt}
Let $ (X_i,Y_i)$, $i \in \mathbb{Z}$, be a promptly $\psi$-mixing process, then for all $\epsilon>0$
\begin{align*}
\lim_{n\to\infty}\Pr(H_n>1-\epsilon)=\mathcal{H}_{X|Y} \; ,\\
\lim_{n\to\infty}\Pr(H_n<\epsilon)=1-\mathcal{H}_{X|Y} \; .
\end{align*}
\end{theo}

\begin{theo}
Let $ (X_i,Y_i)$, $i \in \mathbb{Z}$, be a promptly $\psi$-mixing process, then for all $\beta<1/2$
\begin{align*}
\lim_{n\to\infty}\Pr(Z_n<2^{-N^\beta})=1-\mathcal{H}_{X|Y} \; .
\end{align*}
\end{theo}

As is usual in proofs of polarization, we will analyze how the 
entropies and Bhattacharyya parameters evolve in a single recursion of
the polarization transform. That is, when two smaller polarization
blocks are combined to form a larger block.  Due to the dependence
between the combined blocks, we will need to keep track of more
random variables than is required in the analysis of the memoryless
case.  The following shorthand will then be useful:
\begin{align}
\label{eqn:shorthand}
\begin{split}
U_1^N & = X_{1}^{N} \BN \GN \; , \\  
V_1^N & = X_{N+1}^{2N} \BN \GN \; , \\
Q_i   & = (U_1^{i-1}, Y_1^N) \; ,  \\
R_i   & = (V_1^{i-1}, Y_{N+1}^{2N}) \; . 
\end{split}
\end{align}

\section{Proof of Theorem~\ref{thm:weak}} 
Throughout this section, we assume that $(X_i,Y_i)$, $i \in \mathbb{Z}$, is a promptly $\psi$-mixing process. We will prove Theorem~\ref{thm:weak} by showing that $H_n$ converges almost surely (a.s.) and in $L^1$ to a $\{0,1\}$-valued random variable $H_\infty$. As in \cite{Arikan2009}, we first show that $H_\infty \in [0,1]$. 
\begin{lemm}
    \label{lemm:supermartingale}
    The sequence $H_n$ converges a.s. and in $L^1$ to a random variable $H_\infty \in [0,1]$.
\end{lemm}
\begin{IEEEproof}
    Recall that for $i = 1 + (B_1\dotsc B_n)_2$ we have that
$$
H_n=H(U_i|U_1^{i-1},Y_1^N)=H(U_i|Q_i) \; .
$$ 
Also, for $i$ as above, 
\begin{align*}
H_{n+1}=
\begin{cases}
H(U_i+V_i|Q_i,R_i)
	&\text{if } B_{n+1}=0 \; ,\\
H(V_i|Q_i,R_i,U_i+V_i)
	&\text{if } B_{n+1}=1 \; .
\end{cases}
\end{align*}
Next, note that
\[
H(U_i+V_i|Q_i,R_i)
	+H(V_i|Q_i,R_i,U_i+V_i) 
=H(U_i,V_i|Q_i,R_i) \\
\le H(U_i|Q_i) + H(V_i|R_i) = 2H(U_i|Q_i) \; ,
\]
where the inequality follows since conditioning reduces entropy, and the last step follows from stationarity. Thus, since $B_{n+1}$ is uniform, $E[H_{n+1}|H_1,\dotsc,H_n]\le H_n$. 
The entropy is bounded, $H_n\in[0,1]$, and thus it follows that $H_1,H_2,\dotsc$ is a bounded
supermartingale. We conclude by \cite[Theorem 9.4.5]{Chung:01b} 
that it converges almost surely  and in $L^1$  to a $[0,1]$-valued
random variable~$H_\infty$. 
\end{IEEEproof}

Our approach to proving that $H_\infty \in \{0,1\}$ shares similarities with the proof in \cite[Section 2.2]{Sasoglu2012} for the memoryless case. In essence, the proof there hinges on \cite[Lemma 2.2]{Sasoglu2012}, which shows that if $H(U_i|V_i)$ is bounded away from both $0$ and $1$, then $H(U_i+V_i|Q_i,R_i) - H(U_i|Q_i)$ is bounded away from $0$. Informally, if $H_n$ has not polarized, then it has not converged. Thus, our main focus now is on $H(U_i+V_i|Q_i,R_i)$.  

Recalling the definitions of $Q_i$ and $R_i$ in (\ref{eqn:shorthand}), we see that $Y_N \in Q_i$ and $Y_{N+1} \in R_i$. Since $Y_N$ and $Y_{N+1}$ are generally dependent, we deduce that $Q_i$ and $R_i$ are generally dependent as well. However, suppose that $U_i$ and $V_i$ were independent given $Q_i$ and $R_i$. This is not generally true, but if it were, we would be closer to the memoryless setting and our task of analyzing $H(U_i+V_i|Q_i,R_i)$ would be simpler. Informally, inequality (\ref{eq:UV-almost-independent}) in the next lemma shows that this is ``almost true''.

\begin{lemm}
\label{lem:almost-independent}
For any $\epsilon>0$, the fraction of
indices~$i$ for which
\begin{align}
I(U_i;V_i|Q_i,R_i)&<\epsilon \label{eq:UV-almost-independent} \; , \\
I(U_i;R_i|Q_i)&<\epsilon \label{eq:UR-almost-independent} \; , \\
I(V_i;Q_i|R_i)&<\epsilon\label{eq:VQ-almost-independent} \; ,
\end{align}
approaches~$1$ as $N\to\infty$.
\end{lemm}

\begin{IEEEproof}
We only prove the first and the third inequalities, the second follows by symmetry.
We have
\begin{align*}
\log(\psi_0)
	&\ge E\left[\log
\frac{p_{X_1^{2N},Y_1^{2N}}(X_1^{2N},Y_1^{2N})}{p_{X_1^N,Y_1^N}(X_1^N,Y_1^N)\cdot
p_{X_{N+1}^{2N},Y_{N+1}^{2N}}(X_{N+1}^{2N},Y_{N+1}^{2N})}\right]\\
&=I(X_1^N,Y_1^N; X_{N+1}^{2N},Y_{N+1}^{2N})\\
&=I(U_1^N,Y_1^N; V_{1}^{N},Y_{N+1}^{2N})\\
&=I(Y_1^N; V_{1}^{N},Y_{N+1}^{2N}) + I(U_1^N; V_{1}^{N},Y_{N+1}^{2N}|Y_1^N)\\
&\ge I(U_1^N; V_1^N,Y_{N+1}^{2N}|Y_1^N)\\
&=\sum_{i=1}^N
	I(U_i; V_1^N,Y_{N+1}^{2N}|Y_1^N,U_1^{i-1}) \\
&=\sum_{i=1}^N
I(U_i;  R_i, V_i, V_{i+1}^N|Q_i),
\end{align*}
The first inequality above follows from the definition of $\psi_0$.
Since all terms inside the last sum are non-negative, it follows that at
most $\sqrt{\log(\psi_0)N}$ (a vanishing fraction) of them 
are at least $\sqrt{\log(\psi_0)/N}$ (a vanishing quantity). Thus, to conclude the proof, it suffices to show that the $i$th term is 
greater than both $I(U_i;R_i|Q_i)$ and $I(U_i;V_i|Q_i,R_i)$. Indeed, 
\[
I(U_i; R_i, V_i,   V_{i+1}^N|Q_i) = I(U_i; R_i |Q_i) + I(U_i; V_i |Q_i,R_i) + I(U_i;  V_{i+1}^N|Q_i,V_i, R_i ),
\]
and all the terms are non-negative.
\end{IEEEproof}
 
In fact (\ref{eq:UV-almost-independent}) is the only inequality we will need from Lemma~\ref{lem:almost-independent}. We have stated (\ref{eq:UR-almost-independent}) and (\ref{eq:VQ-almost-independent}) to serve as motivation for the following. Namely, for $1 \leq i \leq N$, we now introduce the random variables $\tilde{U}_i$ and $\tilde{V}_i$. The joint distribution of $(X_1^{2N}, Y_1^{2N}, U_1^N, V_1^N, Q_1^N, R_1^N, \tilde{U}_1^N, \tilde{V}_1^N)$ is defined as follows. First $X_1^{2N}$ and $Y_1^{2N}$ are picked according to the process distribution. This uniquely determines the values of $U_1^N, V_1^N, Q_1^N$, and $R_1^N$, according to  (\ref{eqn:shorthand}). Finally, for each $i=1,2,\ldots,n$ we pick $\tilde{U}_i$ and $\tilde{V}_i$ independently according to the marginal distributions $p_{U_i|Q_i}(\cdot |q_i)$ and $p_{V_i|R_i}( \cdot |r_i)$, where $q_i$ and $r_i$ are the realizations of $Q_i$ and $R_i$. The key property to note is that the joint distribution of $(\tilde{U}_i,\tilde{V}_i)$ with $(Q_i,R_i)$ is of the form
\begin{equation} 
\label{eq:uvtilde}
p_{\tilde{U}_i,\tilde{V}_i,Q_i,R_i}(\tilde{u}_i,\tilde{v}_i,q_i,r_i)\\
=p_{U_i|Q_i}(\tilde{u}_i|q_i)
p_{V_i|R_i}(\tilde{v}_i|r_i)
	 p_{Q_i,R_i}(q_i,r_i) \; .
\end{equation}
Thus, by definition, $\tilde{U}_i$ and $\tilde{V}_i$ are independent given $Q_i$ and $R_i$. In fact, more is true: 
if we replace $U_i$ and $V_i$ by $\tilde{U}_i$ and
$\tilde{V}_i$, respectively, in (\ref{eq:UV-almost-independent})--(\ref{eq:VQ-almost-independent}), then all the mutual informations become
zero. See (\ref{eq:tildeUIndepenedentOfRiGivenQi})--(\ref{eq:tildeUtildeVIndependentGivenQiRi}) in the appendix for a proof of this fact.

%
%
As explained, it will be easier to analyze $H(\tilde{U}_i+\tilde{V}_i|Q_i,R_i)$ in place of $H(U_i+V_i|Q_i,R_i)$. The following corollary to Lemma~\ref{lem:almost-independent} serves as justification for this shift, since it shows that the two quantities are ``close''. It is proved in the appendix and will be used later on.  

\begin{coro}
\label{coro:EntropyOfXorDifference}
For any $\epsilon > 0$, the fraction of indices $i$ for which
\begin{equation}
    \label{eq:entropiesEpsilonClose}
|H(\tilde{U}_i+\tilde{V}_i|Q_i,R_i)-H(U_i+V_i|Q_i,R_i)| < \epsilon
\end{equation}
approaches $1$ as $N \to \infty$.
\end{coro}

Note that by (\ref{eq:uvtilde}), 
\begin{equation}
    \label{eq:UiTildeAlone}
H(\tilde{U_i} | Q_i,R_i) = H(\tilde{U_i} | Q_i) = H(U_i|Q_i) \; .
\end{equation}
Thus, in light of this and Corollary~\ref{coro:EntropyOfXorDifference}, we will consider $H(\tilde{U}_i+\tilde{V}_i|Q_i,R_i) - H(\tilde{U}_i|Q_i)$ as a proxy for our ultimate quantity of interest, $H(U_i+V_i|Q_i,R_i) - H(U_i|Q_i)$. Note that in order to save space, we will usually prefer writing $H(\tilde{U_i} | Q_i)$ in place of the longer but more informative $H(\tilde{U_i} | Q_i,R_i)$. The same remark applies to $H(\tilde{V_i} | Q_i)$ versus $H(\tilde{V_i} | Q_i,R_i)$, which are also equal due to  (\ref{eq:uvtilde}).

Recall that we aim to mimic the memoryless proof in \cite[Section 2.2]{Sasoglu2012} as much as possible. Hence our informal strategy will soon be the following: show that if $H(\tilde{U}_i|Q_i)$ is bounded away from both $0$ and $1$, then  $H(\tilde{U}_i+\tilde{V}_i|Q_i,R_i) - H(\tilde{U}_i|Q_i)$ is bounded away from $0$.

We now motivate the following lemma. Namely, we will now introduce an apparent difficulty, which the following lemma will resolve. Recall that we prefer analyzing $\tilde{U}_i$ and $\tilde{V}_i$ over  $U_i$ and $V_i$, since the former are independent given $(Q_i,R_i)$. In contrast, as we have already mentioned, $Q_i$ and $R_i$ are generally dependent.
This presents an apparent problem with the strategy outlined in the previous paragraph: suppose $H(\tilde{U}_i|Q_i)$ is bounded away from both $0$ and $1$. Suppose further that for every value $q_i$ that $Q_i$ can take, we have that
$H(\tilde{U}_i|Q_i = q_i)$ is either $0$ or $1$. That is, imagine what is effectively an erasure channel, mapping $\tilde{U}_i$ to $Q_i$.
By stationarity, the same property must hold for $H(\tilde{V}_i|R_i = r_i)$.
Now, since $Q_i$ and $R_i$ are \emph{not}
independent, it is conceivable that they collude, i.e., that it is always the case that the values $q_i$ and $r_i$ that the random variables $Q_i$ and $R_i$ respectively take are such that either $H(\tilde{U}_i|Q_i = q_i) = H(\tilde{V}_i|R_i = r_i) = 0$ or
$H(\tilde{U}_i|Q_i = q_i) = H(\tilde{V}_i|R_i = r_i) = 1$. In other words, in two consecutive uses of the above channel, we always have either two non-erasures or two erasures. In such a case, it is easy to see that $H(\tilde{U}_i+\tilde{V}_i|Q_i,R_i) - H(\tilde{U}_i|Q_i)$ is identically $0$. That is, if the above assumptions are valid, our plan is doomed to fail: we have an apparent counter-example in which $H(\tilde{U}_i|Q_i)$ is bounded away from both $0$ and $1$, yet the difference $H(\tilde{U}_i+\tilde{V}_i|Q_i,R_i) - H(\tilde{U}_i|Q_i)$ is not bounded away from $0$. Informally, an important corollary of the following lemma is that such synchronized erasures cannot happen. That is, as intuition for the following lemma, think of $A=1$ ($B = 1$) as indicating that $Q_i$ ($R_i$) corresponds to an erasure of $\tilde{U}_i$ ($\tilde{V}_i$).

\begin{lemm}
\label{lemm:nostuck}
For all $\xi>0$, there exists $N_0$~and~$\delta(\xi)>0$ such that
for all~$N>N_0$
and all $\{0,1\}$-valued random variables 
$A=f(X_1^N,Y_1^N)$ and 
$B=f(X_{N+1}^{2N},Y_{N+1}^{2N})$,
$$
p_A(1)\in(\xi,1-\xi)
	\quad\text{implies}\quad
	p_{A,B}(1,0)>\delta(\xi) \; .
$$
\end{lemm}

\begin{IEEEproof}
Let us start by explaining
informally why the claim is true.  Define
$C=f(X_{2N+1}^{3N},Y_{2N+1}^{3N})$, and suppose to the contrary that
$B$ equals $A$ with very high probability. Hence, by
stationarity, $C$ equals $B$ with very high probability. We conclude that $A$ equals $C$ with probability very close to $1$, a contradiction to the mixing property.

Let us now give a formal proof. First, clearly, we may assume that $\xi \leq 1/2$, or else the claim is vacuous. We have
\begin{align*}
2p_{A,B}(1,0)&=p_{A,B}(1,0)+p_{B,C}(1,0)\\
&\ge p_{A,B,C}(1,0,0)+p_{A,B,C}(1,1,0)\\
&=p_{A,C}(1,0)\\
&=p_A(1)-p_{A,C}(1,1)\\
&\ge p_A(1)(1-\psi_Np_C(1))\\
&= p_A(1)(1-\psi_Np_A(1))
\end{align*}
where the first and last equalities are due to stationarity. Recall that $\psi_N$ converges to $1$ from above. We now commit to an $N_0$ such that $\psi_N < \frac{1}{\sqrt{1-\xi}}$ for all $N > N_0$. Recalling that $p_A(1)\in(\xi,1-\xi)$, we can bound the last term in the above displayed equation as 
\begin{multline*}
    p_A(1)(1-\psi_N p_A(1)) > p_A(1)\left(1- \frac{1}{\sqrt{1-\xi}}p_A(1) \right) > p_A(1) \left(1- \frac{1}{\sqrt{1-\xi}}(1-\xi) \right) \\
    = p_A(1)\left(1-\sqrt{1-\xi}\right) >  \xi\left(1-\sqrt{1-\xi}\right) \; ,
\end{multline*}
assuming that $N > N_0$. That is, for all $N > N_0$, we deduce that $2p_{A,B}(1,0) > \xi\left(1-\sqrt{1-\xi}\right)$. Thus, we take $\delta(\xi) = \xi\left(1-\sqrt{1-\xi}\right)/2$.
%
\end{IEEEproof}

The next lemma will be instrumental in the following setting. Let $q_i$ and $r_i$ be given. Assume that $H(\tilde{U}_i|Q_i = q_i)$ and $H(\tilde{V}_i|R_i = r_i)$ are not both close to $0$, nor are they both close to $1$. To emphasize: we only rule out the case where both entropies are close to each other and extremal. Then, we will deduce from the following lemma that $H(\tilde{U}_i + \tilde{V}_i|Q_i = q_i, R_i = r_i)$ is non-negligibly greater than the mean of $H(\tilde{U}_i|Q_i = q_i)$ and $H(\tilde{V}_i|R_i = r_i)$. The proof is given in the appendix.
\begin{lemm}
\label{lem:trivial}
Let $A$~and~$B$ be independent binary random variables. For every
$\xi>0$, there exists $\Delta(\xi)>0$ such that
\[
\max\{H(A),H(B)\}>\xi\quad \text{ and } \quad
\min\{H(A),H(B)\}<1-\xi
\]
imply
$$
H(A+B)>\frac{H(A)+H(B)}{2}+\Delta(\xi).
$$
\end{lemm} 


We are now ready to state and prove the cardinal lemma of this section. Informally, we now show that if $H_n = H(\tilde{U}_i|Q_i)$ has not polarized, then it has not converged.
\begin{lemm}
    \label{lemm:cardinalTilde}
    For all $\xi > 0$ there exist $\theta(\xi) > 0$ and $N_0$ such that for all $N > N_0$ and all $1 \leq i \leq N$,
\begin{equation}
\label{eqn:h-increase-3}
H(\tilde{U}_i|Q_i)\in(3\xi,1-3\xi)
\text{ implies }
H(\tilde{U}_i+\tilde{V}_i|Q_i,R_i)-H(\tilde{U}_i|Q_i)> 2\theta(\xi) \; .
\end{equation}
\end{lemm}

\begin{IEEEproof}
For a given $\xi > 0$, let $\theta(\xi) = \delta(\xi) \Delta(\xi)/2$, where $\delta(\xi)$ and $\Delta(\xi)$ are as in Lemmas~\ref{lemm:nostuck} and \ref{lem:trivial}. Also, let $N_0$ be as in Lemma~\ref{lemm:nostuck}. The motivation for these choices will soon become apparent. Set $N > N_0$ and let $i$ be given. We must show that $(\ref{eqn:h-increase-3})$ holds. 

Let us first introduce some notation. Let $X$ and $Y$ be generic random variables in this paragraph. Note that $H(X|Y=y)$ is a function of $y$, which we denote in this paragraph as $g(y)$. We shall denote $g(Y)$ as $H(X|\underline{Y})$. We emphasize: the underline in $H(X|\underline{Y})$ signifies that we are dealing with a random variable, which is a function of the underlined quantity.\footnote{One might benefit from verbalizing $H(X|\underline{Y})$ as ``the conditional entropy of $X$, as a function of $Y$''. Note that this definition is similar to the definition of $E[X|Y]$, which is usually taken to be a random variable that is a function of $Y$.} A simple and concise result of this definition is that
    \[
        H(X|Y) = E[H(X |\underline{Y})] \; .
    \]

    Assume that
    \begin{equation}
        \label{eq:nonTrivialEntropy}
        H(\tilde{U}_i|Q_i)\in(3\xi,1-3\xi) \; ,
    \end{equation}                                                               
    otherwise the claim is vacuous.
    Together with our assumption that $\xi$ is positive, the above trivially implies that
    \begin{equation}
        \label{eq:xiBounds}
        0 < \xi < \frac{1}{6} \; .
   \end{equation}
   Recall that $H(\tilde{U}_i|Q_i)= E(\HU)$. In order to keep the notation light, we further denote
    \begin{IEEEeqnarray}{rCl}
        \alpha & =&  \Pr(\HU \leq \xi) \; ,\\
        \beta & = &  \Pr(\HU \in (\xi,1-\xi))\label{eq:betaDef} \; , \\
        \gamma & = &  \Pr(\HU \geq 1-\xi)) \label{eq:gammaDef} \; .                                                                                                                     
    \end{IEEEeqnarray}
    We will prove (\ref{eqn:h-increase-3}) for two cases, $\beta < \xi$ and $\beta \geq \xi$.
    
    \emph{Case 1}: Consider first the case in which
    \begin{equation}
        \label{eq:HUTildeGivenQModerateImprobable}
        \beta  < \xi \; .
    \end{equation}                                                                                                
    In words: the probability that $Q_i$ equals a value $q_i$ for which $H(\tilde{U}_i|Q_i=q_i) \in (\xi,1-\xi)$ is denoted $\beta$, and is less than $\xi$. Informally, for $\xi > 0$ small, this means that a typical realization of $Q_i$ implies either an ``almost certainty'' regarding the value of $\tilde{U}_i$ or an ``almost erasure''.
    
    Informally, we next show that for $\xi$ ``small'', and  under the assumptions (\ref{eq:nonTrivialEntropy}) and (\ref{eq:HUTildeGivenQModerateImprobable}), the probability of an ``almost erasure'', $\gamma$, is not trivial. That is, for a lower bound on $\gamma$, we employ (\ref{eq:nonTrivialEntropy})--(\ref{eq:HUTildeGivenQModerateImprobable}) and deduce that
\begin{IEEEeqnarray*}{rCl}
3 \xi <  H(\tilde{U}_i|Q_i) & \leq & \alpha \cdot  \xi + \beta \cdot (1-\xi) + \gamma \cdot 1 \\
                             & < & \alpha \cdot  \xi + \xi \cdot (1-\xi) + \gamma \cdot 1 \\
                             & \leq &  (1-\gamma) \cdot  \xi + \xi \cdot (1-\xi) + \gamma \cdot 1 \; ,
\end{IEEEeqnarray*}
    where the last inequality follows from $\alpha \leq 1 - \gamma$ (since $\alpha$, $\beta$ and $\gamma$ are probabilities summing to $1$). Rearranging the above gives
    \begin{equation}
        \label{eq:gammaLowerBound}
        \gamma > \frac{\xi + \xi^2}{1-\xi} \; .
    \end{equation}
    For an upper bound on $\gamma$, we again use (\ref{eq:nonTrivialEntropy})--(\ref{eq:HUTildeGivenQModerateImprobable}) to show that
    \[
        1 - 3\xi > H(\tilde{U}_i|Q_i) \geq \alpha \cdot 0 + \beta \cdot \xi + \gamma \cdot (1-\xi) \geq \gamma \cdot (1-\xi)\; .
    \]
    Rearranging gives
    \begin{equation}
        \label{eq:gammaUppedBound}
        \gamma < \frac{1-3\xi}{1-\xi} \; .
    \end{equation}
    By (\ref{eq:xiBounds}), (\ref{eq:gammaLowerBound}), (\ref{eq:gammaUppedBound}), and some simple algebra, we deduce that
    \begin{equation}
        \label{eq:gammaGood}
        \gamma \in (\xi,1-\xi) \; .
    \end{equation}
    
    Recall that by (\ref{eqn:shorthand}), $Q_i$ is a deterministic function of $X_1^N$ and $Y_1^N$. Thus, there clearly exists a $\{0,1\}$-valued function $f$ such that $f(X_1^N,Y_1^N)$ equals $1$ iff  $\HU \geq 1-\xi$. That is, for $\xi$ ``small'', $f(X_1^N,Y_1^N)$ equals $1$ iff $Q_i$ corresponds to an ``almost erasure'' of $\tilde{U}_i$. By the symmetry of definitions in (\ref{eqn:shorthand}) and (\ref{eq:uvtilde}), the above $f$ also satisfies that $f(X_{N+1}^{2N},Y_{N+1}^{2N}) = 1$  iff  $\HV \geq 1-\xi$. Recalling (\ref{eq:gammaDef}), (\ref{eq:gammaGood}), and our definition of $f$, we get from Lemma~\ref{lemm:nostuck} that 
    \begin{equation}
        \label{eq:nonTrivialEntropyPairProbable}
        \Pr\Big(\HU \geq 1-\xi\; , \; \HV< 1-\xi\Big)> \delta(\xi) \; .
    \end{equation}
    Let us now define the ``good'' (with respect to Lemma~\ref{lem:trivial}) set $G$ of pairs $(q_i,r_i)$ as
    \[
        G = \left\{ (q_i,r_i) : 
            \max\{H(\tilde{U}_i|Q_i=q_i),H(\tilde{V}_i|R_i=r_i)\}>\xi\quad \text{ and } \quad \min\{H(\tilde{U}_i|Q_i=q_i),H(\tilde{V}_i|R_i=r_i)\}<1-\xi \right\} \; .
    \]
    By (\ref{eq:xiBounds}) and (\ref{eq:nonTrivialEntropyPairProbable}), 
\begin{equation}
\label{eqn:minmax}
\Pr((Q_i,R_i) \in G)>\delta(\xi) \; .
\end{equation}
We are now ready to show (\ref{eqn:h-increase-3}). We claim that
\begin{multline}
    \label{eq:entroyDiffAsSum}
    H(\tilde{U}_i+\tilde{V}_i|Q_i,R_i)-H(\tilde{U}_i|Q_i) =  H(\tilde{U}_i+\tilde{V}_i|Q_i,R_i)-\frac{H(\tilde{U}_i|Q_i)+H(\tilde{V}_i|Q_i)}{2} \\
    = \sum_{(q_i,r_i)} p_{Q_i,R_i}(q_i,r_i) \left[ H(\tilde{U}_i+\tilde{V}_i|Q_i=q_i,R_i=r_i)-\frac{H(\tilde{U}_i|Q_i=q_i)+H(\tilde{V}_i|R_i=r_i)}{2} \right] \; , \\
    \geq \sum_{(q_i,r_i) \in G} p_{Q_i,R_i}(q_i,r_i) \left[ H(\tilde{U}_i+\tilde{V}_i|Q_i=q_i,R_i=r_i)-\frac{H(\tilde{U}_i|Q_i=q_i)+H(\tilde{V}_i|R_i=r_i)}{2} \right] \; , \\
    > \delta(\xi) \cdot \Delta(\xi) \; .
\end{multline}
Indeed, the first equality is by stationarity; the first inequality is because the term in brackets is always non-negative\footnote{Note that $H(\tilde{U}_i+\tilde{V}_i|Q_i=q_i,R_i=r_i) \geq H(\tilde{U}_i+\tilde{V}_i|\tilde{V}_i,Q_i=q_i,R_i=r_i) = H(\tilde{U}_i|Q_i=q_i)$, and we can similarly lower bound by $H(\tilde{V}_i|R_i=r_i)$.}; the last inequality is by Lemma~\ref{lem:trivial} and (\ref{eqn:minmax}). Thus, recalling that we have taken $\theta(\xi) = \delta(\xi) \Delta(\xi)/2$, we have proved (\ref{eqn:h-increase-3}), under the assumptions (\ref{eq:nonTrivialEntropy}) and (\ref{eq:HUTildeGivenQModerateImprobable}).

\emph{Case 2}: We now aim to prove (\ref{eqn:h-increase-3}), under the assumptions (\ref{eq:nonTrivialEntropy}) and
\begin{equation}
    \label{eq:HUTildeGivenQModerateProbable}
    \beta  \geq  \xi \; .
\end{equation}                                                                                                
This will be shorter, informally because we are now assuming that the probability of $Q_i$ equalling a value for which the entropy of $\tilde{U}_i$ is ``moderate'' is ``sufficiently high''. We start by noticing that under the event $\HU \in (\xi,1-\xi)$ used to define $\beta$ in (\ref{eq:betaDef}), we have that $(Q_i,R_i) \in G$. Thus, the LHS of (\ref{eqn:minmax}) is lower bounded by $\beta$. Next, we claim that $\beta > \delta(\xi)$, and hence (\ref{eqn:minmax}) holds. Indeed, recall from the proof of Lemma~\ref{lemm:nostuck} that $\delta(\xi) = \xi\left(1-\sqrt{1-\xi}\right)/2 < \xi$.
By this and (\ref{eq:HUTildeGivenQModerateProbable}) we deduce that (\ref{eqn:minmax}) holds, and the proof continues as before. Hence, we have proved (\ref{eqn:h-increase-3}), under the assumptions (\ref{eq:nonTrivialEntropy}) and (\ref{eq:HUTildeGivenQModerateProbable}).
\end{IEEEproof}

The following corollary to Lemma~\ref{lemm:cardinalTilde} shifts us back to $U_i$ and $V_i$ from $\tilde{U}_i$ and $\tilde{V}_i$.

\begin{coro}
    \label{coro:cardinal}
For all $\xi > 0$ there exists $\theta(\xi) > 0$ such that
\begin{align}
\label{eqn:h-increase}
H(U_i|Q_i)&\in(3\xi,1-3\xi)
\text{ implies }
H(U_i+V_i|Q_i,R_i)-H(U_i|Q_i)>\theta(\xi)
\end{align}
for a fraction of indices
 $i\in\{1,\dotsc,N\}$ approaching~$1$ as $N \to \infty$.
\end{coro}
\begin{IEEEproof}
    Let $\xi > 0$ be given and take $\theta(\xi)$ as in Lemma~\ref{lemm:cardinalTilde}. Also, take $N_0$ as in Lemma~\ref{lemm:cardinalTilde}. Fix $N > N_0$, and let $\cA$ be the set of indices for which (\ref{eq:entropiesEpsilonClose}) holds, for $\epsilon = \theta(\xi)$. Note that by Corollary~\ref{coro:EntropyOfXorDifference}, the fraction of indices in $\cA$ approaches $1$ as $N \to \infty$. By assumption, for all indices $i$, and specifically for all $i \in \cA$, we have that (\ref{eqn:h-increase-3}) holds. Our aim is to show that (\ref{eqn:h-increase}) holds for all $i \in \cA$ as well.  Indeed, let $i \in \cA$. If $H(U_i|Q_i) \not\in(3\xi,1-3\xi)$, then (\ref{eqn:h-increase}) holds trivially. Thus, assume that $H(U_i|Q_i) \in(3\xi,1-3\xi)$.  By (\ref{eq:UiTildeAlone}), this is equivalent to $H(\tilde{U}_i|Q_i)\in(3\xi,1-3\xi)$. Thus, by assumption, the consequent in (\ref{eqn:h-increase-3}) holds. We deduce that
\begin{IEEEeqnarray*}{rCl}
    \IEEEeqnarraymulticol{3}{l}{H(U_i+V_i|Q_i,R_i)-H(U_i|Q_i)} \\
    \quad & = &  H(U_i+V_i|Q_i,R_i)-H(\tilde{U}_i|Q_i) \\
    & = & H(U_i+V_i|Q_i,R_i) - H(\tilde{U}_i+\tilde{V}_i|Q_i,R_i)+  H(\tilde{U}_i+\tilde{V}_i|Q_i,R_i) -H(\tilde{U}_i|Q_i) \\
    & > & {-}\theta(\xi) + H(\tilde{U}_i+\tilde{V}_i|Q_i,R_i) -H(\tilde{U}_i|Q_i) \\
    & > & {-}\theta(\xi) + 2 \theta(\xi) \\
    & = & \theta(\xi) \; ,
\end{IEEEeqnarray*}
    where the first equality follows from (\ref{eq:UiTildeAlone}); the first inequality follows from (\ref{eq:entropiesEpsilonClose}), recalling that $\epsilon = \theta(\xi)$; and the last inequality follows from our assumption that the consequent in (\ref{eqn:h-increase-3}) holds. Thus, the consequent in (\ref{eqn:h-increase}) holds.
\end{IEEEproof}
    
With Corollary~\ref{coro:cardinal} at hand, the proof of Theorem~\ref{thm:weak} is forthcoming. Indeed, we now essentially repeat the arguments in \cite{Arikan2009}.

\begin{IEEEproof}[Proof of Theorem~\ref{thm:weak}]
    Recall that in Lemma~\ref{lemm:supermartingale}, we proved that $H_n$ converges a.s.\ and in $L^1$ to $H_\infty \in [0,1]$. We next show that $H_\infty$ converges a.s.\ to either $0$ or $1$. That is, we show that for all $\epsilon > 0$, $\Pr(H_\infty \in (\epsilon, 1-\epsilon)) = 0$. Indeed, assume to the contrary that there exists $\epsilon > 0$ for which
    \begin{equation}
        \label{eq:HInftyNotPolarized}
        \Pr(H_\infty \in (\epsilon, 1-\epsilon)) > \rho \; ,
    \end{equation}
    where $\rho > 0$. Next, note that
\begin{IEEEeqnarray*}{rCl}
    \IEEEeqnarraymulticol{3}{l}{\Pr(H_n \in (\epsilon/2,1-\epsilon/2))} \\
    \quad & \geq & \Pr(H_n \in (\epsilon/2,1-\epsilon/2) \quad \mbox{and} \quad  |H_n - H_\infty| < \epsilon/2) \\
    & \geq & \Pr(H_\infty  \in (\epsilon,1-\epsilon) \quad \mbox{and} \quad  |H_n - H_\infty| < \epsilon/2) \\
    & = & \Pr(H_\infty \in (\epsilon,1-\epsilon)) - \Pr(H_\infty \in (\epsilon,1-\epsilon) \quad \mbox{and} \quad |H_n - H_\infty| \geq \epsilon/2) \\
    & \geq & \Pr(H_\infty \in (\epsilon,1-\epsilon)) - \Pr(|H_n - H_\infty| \geq \epsilon/2) \\
    & > & \rho - \Pr(|H_n - H_\infty| \geq \epsilon/2) \; ,
\end{IEEEeqnarray*}
    where the last inequality follows from (\ref{eq:HInftyNotPolarized}). Since a.s.\ convergence implies convergence in probability \cite[Theorem 4.1.2.]{Chung:01b}, we deduce from the above that
    \[
        \liminf_{n \to \infty} \Pr(H_n \in (\epsilon/2, 1-\epsilon/2)) \geq  \rho \; .
    \]
    Now, take $\xi$ such that $3\xi = \epsilon/2$. We deduce from Corollary~\ref{coro:cardinal} that for $n$ large enough,
    \[
        \Pr\left(|H_{n+1} - H_n| > \theta(\xi) \right) > \frac{\rho}{4} \; .
    \]
    However, this implies that $H_n$ cannot converge in probability to $H_\infty$, a contradiction to what was stated earlier. We have proven that $H_\infty \in \{0,1\}$ a.s. 

    We now show that
    \begin{equation}
        \lim_{n \to \infty} E[H_n] = E[H_\infty] \; .
    \end{equation}
    Indeed,
    \[
        E[-|H_n - H_\infty|] \leq E[H_n - H_\infty] \leq E[|H_n - H_\infty|] \; ,
    \]
    and by the $L^1$ convergence of $H_n$ to $H_\infty$ and the sandwich property, the limit of the middle term is $0$. By definition, $\lim_{n \to \infty} E[H_n] = \mathcal{H}_{X|Y}$. Hence, since $H_\infty \in\{0,1\}$ a.s., we must have that $\Pr(H_\infty = 1 ) = 1 - \Pr(H_\infty = 0) = \mathcal{H}_{X|Y}$. Recalling that $H_n$ converges in probability to $H_\infty$, the claim in Theorem~\ref{thm:weakAlt} follows. We end by noting that Theorem~\ref{thm:weakAlt} is equivalent to Theorem~\ref{thm:weak}.

\end{IEEEproof}

\section{Proof of Theorem~\ref{thm:strong}}

Like most proofs of the speed of polarization, our proof of
Theorem~\ref{thm:strong} relies on the following result by Ar\i kan
and Telatar~\cite{ArikanTelatar2009}, although we need the more
general form of the result given\footnote{See also \cite{Tal:17.2p} for a simpler proof.}  in~\cite[Lemma 2.3]{Sasoglu2011}.
\begin{lemm}[\cite{ArikanTelatar2009},\cite{Sasoglu2011}]
    If $Z_n$ converges almost surely to a $\{0,1\}$-valued random variable $Z_\infty$ 
and if there exists $K<\infty$ such that
\begin{align}
\label{eqn:z-bound-1}
Z_n\le KZ_{n-1} \; , & \quad\text{if }B_n=0\\
\label{eqn:z-bound-2}
Z_n\le KZ_{n-1}^2 \; ,& \quad\text{if }B_n=1
\end{align}
then
$$
\lim_{n\to\infty} \Pr(Z_n<2^{-2^{n\beta}})=\Pr(Z_\infty=0)
$$
for all $\beta<1/2$. 
\end{lemm}

Recall from the proof of Theorem~\ref{thm:weak} that $H_n$ converges almost
surely  to a $\{0,1\}$-valued random variable.  It then
follows from the relations~\cite[Proposition 2]{Arikan:10c} 
\begin{align*}
    Z(A|B)^2&\le H(A|B) \\
H(A|B)&\le \log(1+Z(A|B))
\end{align*}
that $Z_n$ also converges almost surely to a $\{0,1\}$-valued random variable $Z_\infty$. Indeed, $H_n\to0$ implies $Z_n\to0$ whereas $H_n\to1$ implies $Z_n\to1$. It then
suffices to show that $Z_n$ satisfies inequalities
\eqref{eqn:z-bound-1}~and~\eqref{eqn:z-bound-2}.  

We claim that this is indeed the case with $K=2\psi_0$.  To see this,
let $\hat{X}_1^{2N},\hat{Y}_1^{2N}$ be distributed as
$P_{X_1^NY_1^N}\cdot P_{X_{N+1}^{2N}Y_{N+1}^{2N}}$, and define the
corresponding variables
$\hat{U}_i,\hat{V}_i,\hat{Q}_i,\hat{R}_i$ as
in~\eqref{eqn:shorthand}.
We know from~\cite[Proposition 5]{Arikan2009}
that
\begin{align}
\label{eqn:z-memoryless-1}
Z(\hat{U}_i+\hat{V}_i
	|\hat{Q}_i,
	\hat{R}_i)
&\le 2Z(\hat{U}_i
	|\hat{Q}_i) \; ,\\
\label{eqn:z-memoryless-2}
Z(\hat{V}_i
	|\hat{Q}_i,
	\hat{R}_i,
	\hat{U}_i+\hat{V}_i)
&\le Z(\hat{U}_i
	|\hat{Q}_i)^2 \; .
\end{align}

Now let $(A,B)$
and $(\hat{A},\hat{B})$ be random variables that can be
written as 
\begin{align*}
(A,B)
	&=f(X_1^{2N},Y_1^{2N})\\
(\hat{A},\hat{B})
	&=f(\hat{X}_1^{2N},\hat{Y}_1^{2N})
\end{align*}
for some function $f$.  Observe that the assumption~\eqref{eqn:psi-mixing}
implies $p_{A,B}\le \psi_0 \cdot p_{\hat{A},\hat{B}}$.  Therefore,
for binary $A$ we have
\begin{align}
\notag
Z(A|B)&=2\sum_b \sqrt{p_{A,B}(0,b)p_{A,B}(1,b)}\\
\notag
&\le 2\psi_0\sum_b 
	\sqrt{p_{\hat{A},\hat{B}}(0,b)
	p_{\hat{A},\hat{B}}(1,b)}\\
\label{eqn:z-tilde}
&=\psi_0\cdot Z(\hat{A}|\hat{B}) \; .
\end{align}

Defining $A=U_i+V_i$ and $B=(Q_i,R_i)$ and
combining \eqref{eqn:z-tilde}~with~\eqref{eqn:z-memoryless-1}
implies~\eqref{eqn:z-bound-1} with $K=2\psi_0$.  Similarly, defining
$A=V_i$ and $B=(Q_i,R_i,U_i+V_i)$ and combining
\eqref{eqn:z-tilde}~with~\eqref{eqn:z-memoryless-2}
implies~\eqref{eqn:z-bound-2} with $K=\psi_0$.  This proves
Theorem~\ref{thm:strong} since $\psi_0<\infty$ by assumption.

\section{Proof of Theorem~\ref{thm:periodic}}
Recall that the process we are considering is described in
Figure~\ref{fig:period4}. Let us start by defining the process
exactly. The state of the process at time $t=1,2,\ldots$ is denoted
$S_t$. Each such state has $4$ possible values, $\{0,1,2,3\}$. The
initial state $S_1$ is picked uniformly at random. The value of $S_1$
determines the value of all $S_t$, specifically, $S_t = S_1 + t-1
\pmod 4$. If $S_t \in \{0,1\}$, then $X_t$, the output of the process
at time $t$, is picked uniformly at random from $\{0,1\}$. If $S_t \in
\{2,3\}$, then $X_t$ equals $0$. Recall that for a given $N$, we
have $U_1^N=X_1^N \BN \GN$.

The proof of Theorem~\ref{thm:periodic} is divided into two parts. In
the first part, we consider $H(U_i|U_1^{i-1},S_1 = s_1)$. Namely, we
consider a setting related to, yet distinct from, that of
Theorem~\ref{thm:periodic}: we assume that the initial state $S_1$ is
known to equal the fixed value $s_1$. As we will see, the case $N=8$ is of particular importance. We refer the reader to Table~\ref{tbl:differentiateStates}, which highlights key features of the distribution of $U_1^6$ when $N=8$, for the $4$ possible values of $s_1$. The entry ``$U_6 \perp U_1^5$'' denotes that $U_6$ is independent of $U_1^5$. The correctness of the Table~\ref{tbl:differentiateStates} is easy to validate by using Table~\ref{tbl:BB00}.

\begin{lemm}
\label{lemm:noPolarizationWhenPeriodicGivenState}
Consider the stationary Markov process described in Figure~\ref{fig:period4}. Then, for $N \geq 8$, the following holds.
\[
\label{eq:fullPolarizationPreserved}
\mbox{For all} \;\; \frac{5N}{8} < i \leq \frac{6N}{8} \;\; \mbox{we have that} \;\;
H(U_i|U_1^{i-1},S_1 = s_1) =
\begin{cases}
0 \; , & \mbox{if $s_1 \in \{1,3\}$} \; , \\
1 \; , & \mbox{if $s_1 \in \{0,2\}$} \; .
\end{cases} 
\]
\end{lemm}
\begin{IEEEproof}
The correctness of the lemma is straightforward to validate for $N=8$. Indeed, for $N=8$ we must only consider $i=6$, and the result follows from the last column of Table~\ref{tbl:differentiateStates}. Namely,  for $s_1 \in \{1,3\}$ we have that $U_6$ is a function of $U_1^5$; for $s_1 \in \{0,2\}$ we have that $U_6$ is independent of $U_1^5$ and is distributed $\Ber(1/2)$.

The general result is proved by induction on $N$. We have proved the basis $N=8$ above. In order to prove the step, let us first tailor the notation (\ref{eqn:shorthand}) to our needs:
\begin{align}
\label{eqn:shorthandPeriodic}
\begin{split}
U_1^N & = X_{1}^{N} \BN \GN\\  
V_1^N & = X_{N+1}^{2N} \BN \GN \\
Q_i   & = U_1^{i-1} \\
R_i   & = V_1^{i-1}
\end{split}
\end{align}
Proving the step is equivalent to proving that for all $\frac{5N}{8} < i \leq \frac{6N}{8}$,
\begin{equation}
\label{eq:UiViTwoTransforms}
H(U_i+V_i|Q_i,R_i,S_1=s_1) = H(V_i|U_i+V_i,Q_i,R_i,S_1=s_1) = H(U_i|Q_i,S_1=s_1) \; .
\end{equation}

Recall that $N$ is a power of $2$ and $N \geq 8$. Thus, $N$ is a multiple of $4$. Since the period of the process is $4$, we have that $S_1 = s_1$ iff $S_{N+1} = s_1$. Moreover, it is easily seen that given that $S_1 = s_1$, $(U_i,Q_i)$ and $(V_i,R_i)$ are identically distributed. Hence,
\begin{equation}
\label{eq:HUiViEqual}
H(U_i|Q_i,S_1=s_1) = H(V_i|R_i,S_{N+1}=s_1) = H(V_i|R_i,S_1=s_1) \; .
\end{equation} 
Moreover, it is easily seen that given that $S_1 = s_1$, $(U_i,Q_i)$ and $(V_i,R_i)$ are independet.

We now prove (\ref{eq:UiViTwoTransforms}) for the two cases of interest. Indeed, if $H(U_i|Q_i,S_1=s_1) = 0$ then $U_i$ and $V_i$ are deterministic function of $Q_i$ and $R_i$, respectively, given that $S_1 = s_1$. Hence, the two equalities in (\ref{eq:UiViTwoTransforms}) follow easily. If $H(U_i|Q_i,S_1=s_1) = 1$, then by (\ref{eq:HUiViEqual}) and the independence of $(U_i,Q_i)$ and $(V_i,R_i)$ given $S_1=s_1$ we deduce that
\begin{multline*}
  2 = H(U_i|Q_i,R_i,S_1 = s_1) + H(V_i|U_i,Q_i,R_i,S_1 = s_1) = H(U_i,V_i|Q_i,R_i,S_1 = s_1) \\ 
  = H(U_i+V_i,V_i|Q_i,R_i,S_1 = s_1) = H(U_i+V_i|Q_i,R_i,S_1 = s_1) + H(V_i|U_i+V_i,Q_i,R_i,S_1 = s_1) \; .
  \end{multline*}
Since the two terms on the RHS are at most $1$, they must both equal $1$, proving (\ref{eq:UiViTwoTransforms}) for this case as well.
\end{IEEEproof}

An immediate corollary of Lemma~\ref{eq:fullPolarizationPreserved} is that $H(U_i|U_1^{i-1},S_1) = 1/2$, for $\frac{5N}{8} < i \leq \frac{6N}{8}$. To see this, note that all $4$ states are equally likely as initial states. What remains is to prove that $S_1$ is essentially known from $U_1^{i-1}$.
\begin{lemm}
\label{lemm:initialStateKnown}
Consider the stationary Markov process depicted in Figure~\ref{fig:period4}. Then, there exists an $\epsilon_N$ such that
\begin{equation}
\label{eq:polarizetozero}
\mbox{for all} \;\; \frac{5N}{8} < i \leq \frac{6N}{8} \;\; \mbox{we have that} \;\;
H(S_1 | U_1^{i-1} ) \leq \epsilon_N \; , \quad \mbox{and} \;\; \lim_{N \to \infty} \epsilon_N = 0 \; .
\end{equation}
\end{lemm}

\begin{table}
\[
\begin{array}{rccccccccccccccc}
U_1 =   & X_1 & + & X_2 & + & X_3 & + & X_4 & + & X_5 & + & X_6 & + & X_7 & + & X_8 \\
U_2 =   &     &   &     &   &     &   &     &   & X_5 & + & X_6 & + & X_7 & + & X_8 \\
U_3 =   &     &   &     &   & X_3 & + & X_4 &   &     &   &     &   & X_7 & + & X_8 \\
U_4 =   &     &   &     &   &     &   &     &   &     &   &     &   & X_7 & + & X_8 \\
U_5 =   &     &   & X_2 &   &     & + & X_4 &   &     & + & X_6 &   &     & + & X_8 \\
U_6 =   &     &   &     &   &     &   &     &   &     &   & X_6 &   &     & + & X_8 \\ \hline \hline
        & X_1 &   & X_2 &   & X_3 &   & X_4 &   & X_5 &   & X_6 &   & X_7 &   & X_8 \\ \hline
S_1 = 0 &  B  &   &  B  &   &  0  &   &  0  &   &  B  &   &  B  &   &  0  &   &  0  \\
S_1 = 1 &  B  &   &  0  &   &  0  &   &  B  &   &  B  &   &  0  &   &  0  &   &  B  \\
S_1 = 2 &  0  &   &  0  &   &  B  &   &  B  &   &  0  &   &  0  &   &  B  &   &  B  \\
S_1 = 3 &  0  &   &  B  &   &  B  &   &  0  &   &  0  &   &  B  &   &  B  &   &  0  
\end{array}
\]
\caption{Properties of $U_1^6$ and $X_1^8$ for $N=8$. Upper half: $U_1^6$ as a function of $X_1^8$. Lower half: distribution of $X_1^8$ as a function of the initial state $S_1$. In the lower half, ``$B$'' is short for $\Ber(1/2)$ and ``$0$'' designates a value of zero with probability one.}
\label{tbl:BB00}
\end{table}

\begin{table}
\[
\begin{array}{rccc}
        & (U_2,U_4)       & (U_1,U_3,U_5)    & U_6 \;\; \mathrm{vs.} \;\;  U_1^5 \\ \hline
S_1 = 0 & U_4 = 0         &                  & U_6 \perp U_1^5  \\
S_1 = 1 & \mathrm{i.i.d.} &  U_5 = U_3       & U_6 = U_4 \\
S_1 = 2 & U_4 = U_2       &                  & U_6 \perp U_1^5 \\
S_1 = 3 & \mathrm{i.i.d.} &  U_5 = U_3 + U_1 & U_6 = U_4 + U_2 
\end{array}
\]
\caption{Distribution properties of $U_1^6$ for $N=8$ and the four possible initial states.
}
\label{tbl:differentiateStates}
\end{table}
\begin{IEEEproof}
We start by giving an informal explanation as to why the claim holds. Consider the first two columns of Table~\ref{tbl:differentiateStates}, and suppose we had many i.i.d.\ realizations of $U_1^5$, all with the same initial state $s_1$. Hence, the first column would allow us to distinguish --- with very high probability --- between $s_1 = 0$, $s_1 = 2$, and $s_1 \in \{1,3\}$: 
\begin{itemize}
\item If $s_1 = 0$ then all the realizations of $U_4$ would equal $0$.
\item If $s_1 = 2$, all realizations would satisfy $U_2 = U_4$. In roughly half the realizations we would have $U_4  = 1$, since $U_4 \sim \Ber(1/2)$. Each such realization would rule out the previous case.
\item If $s_1 \in \{1,3\}$ then in roughly a quarter of the realizations we would have $U_4 = 1$ and $U_2 = 0$, since $U_2$ and $U_4$ are i.i.d.\ and $\Ber(1/2)$. Such an outcome would distinguishing this case from the two previous ones.
\end{itemize}

To distinguish between $s_1 = 1$ and $s_1 = 3$, we utilize the second column of Table~\ref{tbl:differentiateStates}. Specifically, in both cases, $U_1 \sim \Ber(1/2)$. Thus, in roughly half of the realizations, $U_1 = 1$, and for each such realization we can distinguish between $s_1 = 1$ in which $U_5 = U_3$ and $s_1 = 3$ in which $U_5 \neq U_3$.

Lastly, we claim that such independent realization of $U_1^5$ can indeed be attained. Specifically, for $N \geq 8$ and $\frac{5N}{8} < i \leq \frac{6N}{8}$, the vector $U_1^{i-1}$ can be used to deduce the first $5$ entries of each vector in the set $\{X_{1+8(j-1)}^{1+8j} \BEight \GEight : 1 \leq j \leq N/8\}$. Note that since the period of the process is $4$, the state at time $1+8(j-1)$ is equal to $s_1$, for all values of $j$. Also, given $S_1$, all the vectors in the above set are independent.

Let us move on to the formal proof. The statistical properties of $U_1^5$ detailed above are easy to validate using Table~\ref{tbl:BB00}. Suppose we have $N/8$ realizations of $U_1^5$, which are i.i.d.\  given $S_1$. The above description suggests an algorithm for guessing the value of $S_1$: 
\begin{itemize}
\item If all the realizations of $U_4$ equal $0$, set $\hat{S}_1 = 0$.
\item Otherwise, if all realizations satisfy $U_2 = U_4$, set $\hat{S}_1 = 2$.
\item Otherwise, if all realizations satisfy $U_5 = U_3$, set $\hat{S}_1 = 1$.
\item Otherwise, set $\hat{S}_1 = 3$.
\end{itemize}

A straightforward calculation shows that the probability of misdecoding $S_1$ goes down to $0$ exponentially in $N$. By Fano's inequality \cite[Theorem 2.10.1]{CoverThomas:06b}, we have that
\[
H(S_1|U_1^{i-1}) \leq h(p_e) + p_e \log_2 4 \; ,
\]
where $p_e$ is the probability of misdecoding. Since $p_e$ tends to $0$, the RHS of the above tends to $0$ as well.

Recall the set $\{X_{1+8(j-1)}^{1+8j} \BEight \GEight : 1 \leq j \leq N/8\}$, and denote by $A$ the vectors obtained by taking the prefix of length $5$ of each vector in the set. Obviously, the vectors in $A$ are i.i.d.\ given $S_1$, and have the same distribution as the $U_1^5$ discussed above. All that remains to prove is that we can deduce $A$ from $U_1^{i-1}$, when $\frac{5N}{8} < i \leq \frac{6N}{8}$. We prove this by induction on $N$. The case $N=8$ is immediate. For the step, let the set $B$ be defined similarly to $A$, but with $j$ ranging as $N/8+1 \leq j \leq N/8+N/8$. The induction step assumes that $A$ can be deduced from $U_1^{i-1}$. Hence, $B$ can be deduced from $V_1^{i-1}$, where we recall the shorthand (\ref{eqn:shorthandPeriodic}). Recalling the definition of the polar transform, we must prove that both $A$ and $B$ can be deduced from either $(U_1^{i-1}+V_1^{i-1},V_1^{i-1})$ or $(U_1^{i-1}+V_1^{i-1},V_1^{i-1}, U_i+V_i)$. Obviously, this is true.
\end{IEEEproof}

The proof of Theorem~\ref{thm:periodic} is now a simple consequence of the above.
\begin{IEEEproof}[Proof of Theorem~\ref{thm:periodic}]
By the chain rule applied in two ways to $H(U_i,S_1|U_1^{i-1})$ we deduce that
\[
H(U_i|U_1^{i-1}) + H(S_1|U_i,U_1^{i-1}) = H(S_1|U_1^{i-1}) + H(U_i|U_1^{i-1},S_1) \; .
\]
As discussed, an immediate consequence of Lemma~\ref{lemm:noPolarizationWhenPeriodicGivenState} is that $H(U_i|U_1^{i-1},S_1) = 1/2$. Thus,
\[
\left| H(U_i|U_1^{i-1}) - 1/2 \right| =  \left|H(S_1|U_1^{i-1}) - H(S_1|U_i,U_1^{i-1})\right| \; .
\]
By Lemma~\ref{lemm:initialStateKnown}, there exists an $\epsilon_N \to 0$ such that
\[
0 \leq H(S_1|U_i,U_1^{i-1}) \leq H(S_1|U_1^{i-1}) \leq \epsilon_N \; .
\]
Hence,
\[
\left| H(U_i|U_1^{i-1}) - 1/2 \right| \leq \epsilon_N \; .
\]
\end{IEEEproof}
\section{Appendix}
\begin{IEEEproof}[Proof of Corollary~\ref{coro:EntropyOfXorDifference}]
    By marginalizing (\ref{eq:uvtilde}) over $\tilde{v}_i$ we deduce that
\begin{equation}
\label{eq:tildeUIndepenedentOfRiGivenQi}
p_{\tilde{U}_i|Q_i,R_i}(\tilde{u}_i|q_i,r_i) = p_{\tilde{U}_i|Q_i}(\tilde{u}_i|q_i) \; .
\end{equation}
Similarly,
\begin{equation}
    p_{\tilde{V}_i|Q_i,R_i}(\tilde{v}_i|q_i,r_i) = p_{\tilde{V}_i|R_i}(\tilde{v}_i|r_i) \; .
\end{equation}
Thus, by (\ref{eq:uvtilde}) and the above we deduce that $\tilde{U}_i$ and $\tilde{V}_i$ are independent given $Q_i$ and $R_i$,
\begin{equation}
\label{eq:tildeUtildeVIndependentGivenQiRi}
p_{\tilde{U}_i,\tilde{V}_i|Q_i,R_i}(\tilde{u}_i,\tilde{v}_i|q_i,r_i) = p_{\tilde{U}_i|Q_i,R_i}(\tilde{u}_i|q_i,r_i) \cdot p_{\tilde{V}_i|Q_i,R_i}(\tilde{v}_i|q_i,r_i) \; .
\end{equation}

Define
\begin{equation}
\label{eq:h2}
h_2(\alpha) = - \alpha \log_2 \alpha - (1-\alpha) \log_2 (1-\alpha) \; .
\end{equation}
We start with the following simple claim: for $\alpha,\beta$ between $0$ and $1$,
\begin{equation}
\label{eq:h2InOutDiff}
|h_2(\beta) - h_2(\alpha)| \leq h_2(|\beta - \alpha|) \; .
\end{equation}
Indeed, assume w.l.o.g.\ that $\beta \geq \alpha$. Then,
\begin{equation}
\label{eq:h2diffleq}
h_2(\beta) - h_2(\alpha) = \int_\alpha^\beta h'_2(t) \, dt \leq \int_0^{\beta - \alpha} h'_2(t) \, dt = h_2(\beta-\alpha) \; ,
\end{equation}
where the inequality follows from the concavity of $h_2$ (the derivative $h'_2$ is decreasing). Similarly,
\begin{equation}
\label{eq:h2diffgeq}
h_2(\beta) - h_2(\alpha) = \int_\alpha^\beta h'_2(t) \, dt \geq \int_{1-(\beta-\alpha)}^1 h'_2(t) \, dt = -h_2(1-(\beta-\alpha)) = -h_2(\beta-\alpha) \; .
\end{equation}
We deduce (\ref{eq:h2InOutDiff}) from (\ref{eq:h2diffleq}) and (\ref{eq:h2diffgeq}).

For $q_i$ and $r_i$ fixed, let us adopt the shorthand $\alpha=p_{U_i+V_i|Q_i,R_i}(0|q_i,r_i)$ and $\beta=p_{\tilde{U}_i+\tilde{V}_i|Q_i,R_i}(0|q_i,r_i)$. We claim that
\begin{IEEEeqnarray}{rCl}
|H(\tilde{U}_i+\tilde{V}_i|Q_i,R_i)-H(U_i+V_i|Q_i,R_i)| & = & \left|\sum_{q_i,r_i} p_{Q_i,R_i}(q_i,r_i) \big(h_2(\beta)-h_2(\alpha) \big)\right| \nonumber \\
& \leq & \sum_{q_i,r_i} p_{Q_i,R_i}(q_i,r_i) |h_2(\beta)-h_2(\alpha)| \nonumber \\
& \leq & \sum_{q_i,r_i} p_{Q_i,R_i}(q_i,r_i) h_2(|\beta- \alpha|) \nonumber \\
\label{eq:HxorHalfBound}& \leq & h_2\left(\sum_{q_i,r_i} p_{Q_i,R_i}(q_i,r_i)|\beta - \alpha|\right) \; .
\end{IEEEeqnarray}
The second inequality follows form (\ref{eq:h2InOutDiff}) while the third inequality follows by applying Jensen's inequality \cite[Theorem 2.6.2]{CoverThomas:06b} with respect to the concave function $h_2$.

Our aim now is to bound the argument of $h_2$ in the RHS of the above displayed equation. Let us use the shorthand $p = p_{U_i,V_i|Q_i,R_i}$ and $\tilde{p} = p_{\tilde{U}_i,\tilde{V}_i|Q_i,R_i}$.
By (\ref{eq:tildeUtildeVIndependentGivenQiRi}),
\[
I(U_i;V_i|Q_i,R_i) = \sum_{q_i,r_i} p_{Q_i,R_i}(q_i,r_i) D(p||\tilde{p}) \; ,
\]
where $D(p||\tilde{p})$ is the relative entropy between $p$ and $\tilde{p}$, for $q_i$ and $r_i$ fixed,
\[
D(p||\tilde{p}) = \sum_{u_i,v_i} p(u_i,v_i|q_i,r_i) \log_2 \frac{p(u_i,v_i|q_i,r_i)}{\tilde{p}(u_i,v_i|q_i,r_i)} \; .
\]
Next, let us denote $p_+ = p_{U_i + V_i | Q_i,R_i}$ and $\tilde{p}_+ = p_{\tilde{U}_i + \tilde{V}_i | Q_i,R_i}$. Obviously, $p_+$ is gotten by quantizing $p$:
\[
p_+(0|q_i,r_i) = p(0,0|q_i,r_i) + p(1,1|q_i,r_i) \; , \quad p_+(1|q_i,r_i) = p(1,0|q_i,r_i) + p(0,1|q_i,r_i) \; .
\]
The same quantization is used to derive $\tilde{p}_+$ from $\tilde{p}$. A simple consequence of the log-sum inequality \cite[Theorem 2.7.1]{CoverThomas:06b} is that such a quantization reduces the relative entropy. Namely, for $q_i,r_i$ fixed,
\[
D(p||\tilde{p}) \geq D(p_+||\tilde{p}_+) \; .
\]
Recalling that $\alpha=p_+(0|q_i,r_i)$ and $\beta=\tilde{p}_+(0|q_i,r_i)$, we get from Pinsker's inequality \cite[Equation 11.147]{CoverThomas:06b} that
\[
D(p_+||\tilde{p}_+) \geq \frac{1}{2 \ln 2} \cdot 2(\beta - \alpha)^2 \; .
\]
Aggregating the above inequalities yields
\[
I(U_i;V_i|Q_i,R_i) \geq \frac{1}{\ln 2} \sum_{q_i,r_i} p_{Q_i,R_i}(q_i,r_i)  \cdot (\beta - \alpha)^2 \; .
\]

Now is the time to invoke Lemma~\ref{lem:almost-independent}. Namely, for an $\epsilon'$ which we will determine shortly, the fraction of indices $i$ for which $I(U_i;V_i|Q_i,R_i) \leq \epsilon'$ approaches $1$ as $N \to \infty$. Thus, for such an index $i$ we have that
\[
\frac{1}{\ln 2} \sum_{q_i,r_i} p_{Q_i,R_i}(q_i,r_i)  \cdot |\beta - \alpha|^2 \leq I(U_i;V_i|Q_i,R_i) \leq \epsilon' \; .
\]
Since squaring is a convex function, we apply Jensen's inequality and deduce that
\[
\sum_{q_i,r_i} p_{Q_i,R_i}(q_i,r_i) \cdot |\beta - \alpha| \leq \sqrt{\epsilon' \cdot \ln 2} \; .
\]
Assuming the RHS of the above is less than $1/2$, we deduce from the above, the monotonicity of $h_2$ in $[0,1/2]$,  and (\ref{eq:HxorHalfBound}) that
\[
|H(\tilde{U}_i+\tilde{V}_i|Q_i,R_i)-H(U_i+V_i|Q_i,R_i)| \leq h_2(\sqrt{\epsilon' \cdot \ln 2}) \; .
\]
Thus, taking $\epsilon'$ small enough so that $\sqrt{\epsilon' \cdot \ln 2} \leq 1/2$ and $h_2(\sqrt{\epsilon' \cdot \ln 2}) \leq \epsilon$ finishes the proof.

\end{IEEEproof}

\begin{IEEEproof}[Proof of Lemma~\ref{lem:trivial}]
Denote the distributions of $A$ and $B$ as
\[
A \sim \Ber(\alpha) \; , \quad B \sim \Ber(\beta) \; .
\]
We will assume w.l.o.g.\ that $0 \leq \alpha \leq \beta \leq 1/2$. Thus, according to our assumptions,
\[
h_2(\alpha) \leq h_2(\beta) \; , \quad h_2(\alpha)  \leq 1 - \xi \; , \quad h_2(\beta) \geq \xi \; ,
\]
where $h_2$ is defined in (\ref{eq:h2}). Since $h_2$ is strictly increasing when restricted to the domain $[0,1/2]$, it is invertible and we conclude that
\[
0 \leq \alpha  \leq h_2^{-1}(1 - \xi) \; , \quad h_2^{-1}(\xi) \leq \beta \leq \frac{1}{2} \; .
\]
We simplify the above to
\begin{equation}
\label{eq:pqtau1minustau}
0 \leq \alpha  \leq \frac{1}{2} - \sigma \; , \quad \sigma \leq \beta \leq \frac{1}{2}  \; .
\end{equation}
where
\[
\sigma = \sigma(\xi) = \min \left\{h_2^{-1}(\xi), \frac{1}{2}-h_2^{-1}(1-\xi) \right\} \; .
\]

Define the random variable $D = (C,T)$ as follows,
\[
D = (C,T) \; , \quad T \sim \Ber(1/2) \; , \quad C =
\begin{cases}
A & \mbox{if $T=0$} \; , \\
B & \mbox{if $T=1$} \; .
\end{cases}
\]
One easily gets that
\[
H(A+B|D) = \frac{H(A) + H(B)}{2} \; .
\]
Thus, we are interested in bounding the difference
\[
H(A+B) - H(A+B|D) = I(A+B;D) \; .
\]
We write $I(X+Y;D)$ as in terms of relative entropy \cite[Equation (2.29)]{CoverThomas:06b}, and lower bound that with Pinsker's inequality \cite[Equation 11.147]{CoverThomas:06b}. Doing so results in a straightforward calculation which yields
\begin{IEEEeqnarray*}{rCl}
\label{eq:differenceBoundedByL1} 
H(A+B) - \frac{H(A) +  H(B)}{2} & \geq & \frac{2}{\ln 2} \big(\beta(1-\beta)|1-2\alpha| + \alpha(1-\alpha)|1-2\beta| \big)^2 \\
& \geq & \frac{2}{\ln 2} \big(\beta(1-\beta)|1-2\alpha| \big)^2 \\
& \geq & \frac{2}{\ln 2} \big(\sigma(1-\sigma)|2\sigma| \big)^2 \\
& = & \frac{8}{\ln 2} \sigma^4(1-\sigma)^2 \; , \\
\end{IEEEeqnarray*}
where the last inequality follows from (\ref{eq:pqtau1minustau}). Now, simply take $\Delta$ as the RHS of the above.


\end{IEEEproof}
\twobibs{
\bibliographystyle{IEEEtran}
\bibliography{polar-memory}

\begin{thebibliography}{10}
\providecommand{\url}[1]{#1}
\csname url@samestyle\endcsname
\providecommand{\newblock}{\relax}
\providecommand{\bibinfo}[2]{#2}
\providecommand{\BIBentrySTDinterwordspacing}{\spaceskip=0pt\relax}
\providecommand{\BIBentryALTinterwordstretchfactor}{4}
\providecommand{\BIBentryALTinterwordspacing}{\spaceskip=\fontdimen2\font plus
\BIBentryALTinterwordstretchfactor\fontdimen3\font minus
  \fontdimen4\font\relax}
\providecommand{\BIBforeignlanguage}[2]{{%
\expandafter\ifx\csname l@#1\endcsname\relax
\typeout{** WARNING: IEEEtran.bst: No hyphenation pattern has been}%
\typeout{** loaded for the language `#1'. Using the pattern for}%
\typeout{** the default language instead.}%
\else
\language=\csname l@#1\endcsname
\fi
#2}}
\providecommand{\BIBdecl}{\relax}
\BIBdecl

\bibitem{Arikan2009}
E.~Ar{\i}kan, ``Channel polarization: A method for constructing
  capacity-achieving codes for symmetric binary-input memoryless channels,''
  \emph{IEEE Trans. Inform. Theory}, vol.~55, pp. 3051--3073, 2009.

\bibitem{Wang+:15c}
R.~Wang, J.~Honda, H.~Yamamoto, R.~Liu, and Y.~Hou, ``Construction of polar
  codes for channels with memory,'' in \emph{Proc. IEEE Inform. Theory Workshop
  (ITW'2015)}, Jeju Island, Korea, 2015, pp. 187--191.

\bibitem{Sasoglu2011}
E.~\c{S}a\c{s}o\u{g}lu, ``Polar coding theorems for discrete systems,'' Ph.D.
  dissertation, {E}cole {P}olytechnique {F}{\'{e}}d{\'{e}}rale de {L}ausanne,
  2011.

\bibitem{ShuvalTal:17.2p}
B.~Shuval and I.~Tal, ``Fast polarization for processes with memory,''
  \emph{Submitted to IEEE Trans. Inform. Theory}, 2017.

\bibitem{Shields1996}
P.~C. Shields, \emph{The Ergodic Theory of Discrete Sample Paths}, ser.
  Graduate Studies in Mathematics.\hskip 1em plus 0.5em minus 0.4em\relax
  Providence (R.I.): American Mathematical Society, 1996, vol.~13.

\bibitem{Bradley:07b}
R.~C. Bradley, \emph{Introduction to Strong Mixing Conditions}.\hskip 1em plus
  0.5em minus 0.4em\relax Heber City, Utah: Kendrick Press, 2007, vol.~I.

\bibitem{CoverThomas:06b}
T.~M. Cover and J.~A. Thomas, \emph{Elements of Information Theory},
  2nd~ed.\hskip 1em plus 0.5em minus 0.4em\relax Wiley, 2006.

\bibitem{Bradley:99p}
R.~C. Bradley, ``Equivalent mixing conditions for {M}arkov chains,''
  \emph{Statis. Probab. Letters}, vol.~41, pp. 97--99, 1999.

\bibitem{Chung:01b}
K.~L. Chung, \emph{A Course in Probability Theory}, 3rd~ed.\hskip 1em plus
  0.5em minus 0.4em\relax San Diego: Academic Press, 2001.

\bibitem{Sasoglu2012}
E.~\c{S}a\c{s}o\u{g}lu, ``Polarization and polar codes,'' in \emph{Found. and
  Trends in Commun. and Inform. Theory}, vol.~8, no.~4, 2012, pp. 259--381.

\bibitem{ArikanTelatar2009}
E.~Ar{\i}kan and E.~Telatar, ``On the rate of channel polarization,'' in
  \emph{Proc. IEEE Int'l Symp. Inform. Theory (ISIT'2009)}, Seoul, South Korea,
  2009, pp. 1493--1495.

\bibitem{Tal:17.2p}
I.~Tal, ``A simple proof of fast polarization,'' \emph{IEEE Trans. Inform.
  Theory}, vol.~63, no.~12, pp. 7617--7619, December 2017.

\bibitem{Arikan:10c}
E.~Ar{\i}kan, ``Source polarization,'' in \emph{Proc. IEEE Int'l Symp. Inform.
  Theory (ISIT'2010)}, Austin, Texas, 2010, pp. 899--903.

\end{thebibliography}
}
{

}
\end{document}